%% file: dementia_engl.tex
\documentclass[a4paper,11pt]{scrartcl}
\usepackage[centertags]{amsmath}
\usepackage{amsfonts}
\usepackage{amssymb}
\usepackage{amsthm}
\usepackage{nicefrac}
\usepackage{titlesec}
\usepackage[ruled, section]{algorithm}
\usepackage{pifont}
\usepackage[ansinew]{inputenc}
\usepackage{graphicx}
\usepackage{natbib}

\theoremstyle{plain}
\newtheorem{thm}{Theorem}[section]

\begin{document}

\title{A new method for deriving incidence rates from prevalence data and its application to
dementia in Germany}

\author{Ralph Brinks\footnote{rbrinks@ddz.uni-duesseldorf.de}\\
Institute for Biometry and Epidemiology\\German Diabetes Center\\
Düsseldorf, Germany}

\date{}

\maketitle

\begin{abstract}
This paper descibes a new method for deriving incidence rates of a chronic disease from 
prevalence data. It is based on a new ordinary differential equation, which relates 
the change in the age-specific prevalence to the age-specific incidence and mortality rates.
The method allows the extraction of longtudinal information from cross-sectional studies.
Applicability of the method is tested in the prevalence of dementia in Germany. The derived age-specific
incidence is in good agreement with published values.
\end{abstract}

\emph{Keywords:} Chronic diseases; Dementia; Incidence; 
Prevalence; Mortality; Compartment model; Cubic spline; Ordinary differential equation.

\input{dementia_engl_10}

\appendix

\label{lastpage}

\end{document}

%% file: dementia_engl_10.tex
\section{Introduction}
Basic epidemiological characteristics of a disease are the
\emph{prevalence}, the proportion of diseased persons in 
the population, and the
\emph{incidence}, which focusses on the number of new cases. Both
characteristics are fundamentally different: the first measures the 
actual presence of the disease, the second refers to the new
cases. Typically, prevalence and incidence of a disease are surveyed 
in observational studies. The prevalence can easily be assessed in 
cross-sectional studies: The study population is interviewed or examined 
with respect to the disease. The classical
approach to measure incidence is the cohort study, which is
somewhat more complex. A certain group of patients is 
examined whether the disease exists at the start of the study. The
healthy individuals of the group will be examined at least once more at a later
point in time, to find out whether the disease occured in the meanwhile. 
Since at least one follow-up investigation must take place, a
cohort study mostly is much more complex and expensive than a cross-sectional study. 
Particular difficulties arise due the fact that participants get lost after the baseline
examination (loss to follow-up).

\bigskip

For some questions, the incidence of a disease is more important than knowing the number 
of those who are already ill. Many of the questions in health services 
research, such as the allocation of resources need
information about the number of expected patients. Within epidemiology, there are several attempts
to derive incidences from prevalence data. A simple, popular example may illustrate this.
Consider a closed population of size $N$; this means there is no migration and the numbers of 
births and deaths are exactly the same for the considered period of length $\Delta t > 0$. Let $C$ denote the
number of persons in the population who suffer from a chronic disease ($C$ stands for
\emph {cases}). Assuming that the number of diseased persons for the time period is constant, 
it follows that the number of new cases is just equal to the number of patients
who die. Hence, it holds
\begin {equation*}
 \left (N - C  \right) \cdot i \cdot \Delta t = C \cdot m_1 \cdot \Delta t,
\end{equation*}
where $i$ is the incidence rate and $m_1$ is the mortality of the diseased\footnote{For 
later use we denote the mortality of the healthy and the diseased with $m_0$ 
and $m_1$, respectively. The subindex dichotomizes the presence of the disease.}.
By defining the (overall) prevalence $p := \tfrac{C}{N},$ the term $\tfrac{C}{N - C}$ 
can be expressed as $\tfrac{C}{N - C} = \tfrac{p}{1-p}$, which is called 
\emph{prevalence odds}. Since the inverse of the mortality $m_1$ is the
mean duration $d$ of the disease, it follows:

\begin{equation*}
\frac{p} {1-p} = i \cdot d.
\end{equation*}

This corresponds to the often found statement that the prevalence odds
equals the product of incidence and disease duration (see for example
\citep{Szk07}). For rare diseases ($1-p \approx 1 $) this reads
as: prevalence equals the product of incidence and duration. 

\bigskip

Beside this simple example, a number of more complex approaches exist
to estimate the incidence from prevalence data. The article by
\citet{Lan99} gives an overview. This work reports about a new method, 
which is based on a simple
compartment model and uses an ordinary
differential equation (ODE) to express transitions between
the compartments. Compartment models in epidemiology 
go back at least until the early 1990s
(see for example \citep{Kei90}). \citeauthor{Mur94}
from the \emph{Harvard Center for Population and Development Studies}
describe that they express the transitions between the compartments in terms of 
ODEs \citeyearpar{Mur94}. Without quoting Keiding, they call their model \emph{Harvard
Incidence-Prevalence Model}. Unfortunately, they do not describe their equations. 
Later they publish another (slightly more complicated) 
compartment model and present the associated ODEs, \citep {Mur96}. 
Our approach builds on the original model of Keiding and a two-dimensional 
system of ODEs. By analytical
transformations this two-dimensional system can be reduced to an
one-dimensional ODE. The reduced equation to our knowledge has not yet been published 
by other groups.  We take this equation further to derive the
incidence rate from prevalence data. In contrast to the multi-dimensional system
of Murray and Lopez, the one-dimensional equation can be solved easily for the 
incidence.

\bigskip

This paper is organized as follows: Section \ref{sec:methods}
describes the newly discovered link between the
age distributions of the prevalence, incidence and
mortality rates. In Section \ref{sec:appl} the new
method is applied to data of the statutory health
insurance (SHI) in Germany. The age distribution of
persons with a diagnosis of dementia is used to derive the incidence
rates in the associated age groups. Finally, the results 
are discussed in section \ref{sec:discussion}.

\section{The new relation beween incidence, prevalence and mortality}\label{sec:methods}
Compartment models are widespread in medicine and other sciences, \citep{God83}. 
In epidemiology of infectious diseases
they play a prominent role with a variety of applications, \citep{Kee07}.
A simple model for the study of non-infectious diseases 
is shown in Figure \ref{fig:3states}. Three states \emph{Normal, Disease}
and \emph{Death} are considered, plus the transitions between
the states, \citep{Kei91, Mur94}. In general, the transition rates depend on the
calendar time $t$ (sometimes called the \emph{period}) and the age $a$. Henceforth,
only irreversible diseases are considered. The transition from the state 
\emph{Disease} to the state \emph{Death} often depends on the duration 
$d$ of the disease. The influence of calendar time reflects, for example, the
change in mortality or medical progress over decades.

\begin{figure*}[ht]
\centerline{\includegraphics[keepaspectratio,
width=14cm]{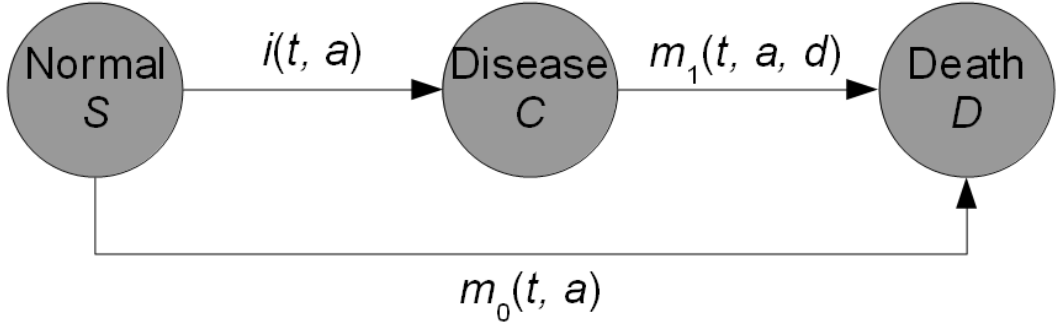}} \caption{Simple model 
of a chronic disease with three states. Persons in the state
\emph{Normal} are healthy with respect to the considered disease.
In the state \emph{Disease} they suffer from the disease. The
transition rates depend on the calendar time $t,$ on the age $a,$ and
in case of the disease-specific mortality $m_1$ also on the
disease's duration $d$.} \label{fig:3states}
\end{figure*}

As shown in Figure \ref{fig:3states} people in the population get the disease with incidence rate $i$. 
The mortality rate depends on the state: Non-diseased and diseased 
persons die with rates $m_0$ and $m_1$, respectively. Mostly, the rate $m_1$ will
be higher than the rate of $m_0$. For historical reasons, the numbers of individuals in the states 
\emph{Normal} and \emph {Disease} are denoted $S$ (susceptibles) and $C$ (cases), respectively.

\noindent Henceforth, we need the following assumptions:
\begin{enumerate}
\item The rates $i, m_0,$ and $m_1$ do not depend on calendar time $t$,
\item the mortality rate $m_1$ of the diseased does not depend on the duration $d$,
\item the population is closed (i.e. there is no migration),
\item the birth-rates of new-borns with and without the disease are constant over time.
\end{enumerate}

Furthermore, let us assume that the changes in $S$ and $C$ are proportional to the differences
of the in- and outflows to and from the compartments:
\begin{equation}\label{eq:MurrayODE}
\begin{split}
    \frac{\mathrm{d} S}{\mathrm{d} a} &= - \left ( i(a) + m_0(a) \right ) \cdot S\\
    \frac{\mathrm{d} C}{\mathrm{d} a} &= i(a) \cdot S - m_1(a) \cdot C.\\
\end{split}
\end{equation}

\bigskip

\noindent With these assumptions the central result of this work can be formulated:
\begin{thm}\label{th:central1}
Let mortalities $m, m_0 \in C^0\left(\left[0, \infty \right)\right)$ and $S, C
\in C^1\left(\left[0, \infty \right)\right)$ with $S(a) + C(a) > 0$ for all $a
\in \left[0, \infty \right),$ then the age-specific prevalence 
\begin{equation}\label{eq:prev}
p = \frac{C}{S+C}
\end{equation}
is differentiable in $\left[0, \infty \right)$ and it holds
\begin{equation}\label{eq:eq4}
\frac{\mathrm{d} p}{\mathrm{d} a} = (1-p) \cdot \left ( i - p \cdot \left (m_1 -
m_0 \right ) \right ).
\end{equation}
\end{thm}
\begin{proof}
This is an easy application of the quotient rule to Eq. \eqref{eq:prev} 
using Eq. \eqref{eq:MurrayODE}.
\end{proof}

Depending on what information is given about the mortalities, 
the ODE \eqref{eq:eq4} changes its type (see Table
\ref{t:Types}). Note that the overall
mortality $m$ in the population can be expressed as
\begin{equation}\label{eq:overall}
m(a) = (1-p(a)) \cdot m_0(a) + p(a) \cdot m_1(a).
\end{equation}
Furthermore, in some cases the relative mortality risk
$R(a) = \tfrac{m_1(a)}{m_0(a)}$ is known.

\begin{table*}[ht]
 \centering
 \def\~{\hphantom{0}}
 \begin{minipage}{140mm}
  \caption{Type of the ODE \eqref{eq:eq4}.} \label{t:Types}
  \begin{tabular*}{\textwidth}{@{}l@{\extracolsep{\fill}}l@{\extracolsep{\fill}}l@{\extracolsep{\fill}}}
  \hline \hline
Known mortality  & Right-hand side    & Type of the ODE\\
\hline
$m, m_0$   & $(1-p) \cdot (i - (m - m_0))$ & Linear \\
$m_0, m_1$ & $(1-p) \cdot (i - p \cdot (m_1 - m_0))$             & Riccati\\
$m_0, R$   & $(1-p) \cdot (i - p \cdot m_0 \cdot (R - 1))$       & Riccati\\
$m_1, R$   & $(1-p) \cdot (i - p \cdot m_1 \cdot (1 - 1/R))$     & Riccati\\
$m, m_1$   & $(1-p) \cdot \left (i - p \cdot \frac{m_1 - m}{1-p} \right )$     & Abel\\
$m, R$     & $(1-p) \cdot \left (i - m \cdot \left( 1 - \left( p \cdot \left(R -
1\right) + 1 \right)^{-1} \right ) \right)$ & Abel\\
\hline
\end{tabular*}
\end{minipage}
\vspace*{-6pt}
\end{table*}

If the ODE is the linear, it can be solved analytically. 
If it is Riccatian or Abelian, a general solution is not known, \citep{Kam83}. 
Since the overall mortality $m$
in many populations can be obtained by official life tables; and relative mortality risks $R$ for 
several diseases often are reported in epidemiological studies, the most important case
is when $m$ and $R$ are given. The following section will show an application of this.

\bigskip

Note, that independence from calendar time $t$, zero migration and 
constant birth rates are
crucial for Eq. \eqref{eq:MurrayODE}. This can be seen by realizing that the population 
size $N(a) = S(a) + C(a)$ at age $a$ fulfills the following equation:
\begin{align*}
\frac{\mathrm{d} N}{\mathrm{d} a} &= \frac{\mathrm{d} S}{\mathrm{d} a} + \frac{\mathrm{d} C}{\mathrm{d} a}\\
                                  &= -m_0 \cdot S - m_1 \cdot C\\
                                  &= - N \cdot \left [ (1-p) \cdot m_0 + p\cdot m_1 \right ].
\end{align*}
Thus, by using Eq. \eqref{eq:overall} it holds $\tfrac{\mathrm{d}
N}{\mathrm{d} a} = - m \cdot N$, which is the defining equation of
a stationary population, \citep{Pre82}. These assumptions ensure 
the population being stationary, \citep[pp. 53ff]{Pre01}.

\section{Application to dementia}\label{sec:appl}
As already described in the introduction, the most important application is 
the derivation of the age-specific incidence rate $i(a)$ from the
age distribution $p(a)$ of the prevalence of a chronic disease. 
The usefulness of the method is examined in an example of dementia. 
Prevalence of dementia in Germany is reported in the work of 
\citeauthor{Zie09} \citeyearpar{Zie09}. Basis for the values published there 
were claims data from the German statutory health insurance (SHI) in the year 2002.
About 90 percent of the whole population in Germany are members of the SHI. 
A three percent random sample of all of these is used for the analysis.
Hence, information of more than 2.3 million people are taken into account.

\smallskip

For each of the persons in the three-percent sample, demographic
data (age, gender), the number of doctor visits and
hospital stays, both with diagnostic positions in
ICD coding, are included. \citeauthor{Zie09} associate the following ICD-10-GM diagnoses 
with dementia: F00, F01, F02, F03, G30. The resulting prevalences in Germany are 
reported in Table \ref{t:Prevalence}.

\begin{table*}[ht]
 \centering
 \def\~{\hphantom{0}}
 \begin{minipage}{100mm}
  \caption{Prevalence of dementia in members of the German SHI in 2002 \citet{Zie09}} \label{t:Prevalence}
  \begin{tabular*}{\textwidth}{@{}l@{\extracolsep{\fill}}l@{\extracolsep{\fill}}
l@{\extracolsep{\fill}}}%
  \hline \hline
Age group (in years)  & Females (\%) & Males (\%)\\ \hline
60--64 & 0.6 & 0.8 \\%
65--69 & 1.3 & 1.5 \\%
70--74 & 3.1 & 3.2 \\%
75--79 & 6.8 & 5.6 \\%
80--84 & 12.8& 10.3 \\%
85--89 & 23.1& 17.9 \\%
90--94 & 31.3& 24.2 \\%
\hline
\end{tabular*}
\end{minipage}
\vspace*{-6pt}
\end{table*}

%\begin{figure}
%\centerline{\includegraphics[width=11.16cm,
%keepaspectratio]{../img/Prevalence.eps}} \caption{Altersverteilung
%der Demenz-Prävalenz bei Frauen (durchgezogene Linie) und Männern
%(gestrichelt).} \label{fig:Prevalence}
%\end{figure}

\bigskip

The prevalence data show that dementia is more frequent in women
aged $\ge$ 75 years than in men in same age group. Unfortunately 
\citeauthor{Zie09} have not reported confidence intervals or \emph{p-}values
to decide whether the differences in the age groups are significant. Due to the
large sample size this is likely.

\bigskip

For the application of the one-dimensional ODE \eqref{eq:eq4} we need
statements about the mortality. Here, we use the general mortality $m$
as surveyed by the Federal Statistical Office of Germany. The relative mortality $R$ 
of persons with dementia can be found in \citep{Rai10}: In the first year after the 
diagnosis of the disease, it is about $3.7$ and in subsequent years about
$ 2.4.$ In this work, the relative mortality is set to be constant at $R(a) = 2.4.$

\bigskip

In order to derive the incidence rate $i(a)$ from Eq. \eqref{eq:eq4}, the following
steps are performed:
\begin{enumerate}
    \item Derive a spline function $s(a)$ that interpolates the prevalence data.
    \item Calculate the derivative $\tfrac{\mathrm{d}s}{\mathrm{d}a}$ and define the function
         $$c = \frac{\mathrm{d}s / \mathrm{d}a}{1 - s}.$$
    \item The incidence rate $i(a)$ can be expressed by $$i(a) = c(a) +
    m(a) \cdot \left( 1 - \left( s(a) \cdot \left(R(a) - 1\right) + 1 \right)^{-1}
    \right).$$
\end{enumerate}
The spline is used to transform the discrete values of the prevalence data
into a differentiable function. Here, the uniquely defined cubic spline
with natural bounding conditions that interpolates the prevalence data is chosen.
It is two times differentiable. Calculations are performed with the statistical software R 
(The R Foundation for Statistical Computing), version 2.12.0.

\bigskip

Using the prevalence data as shown in Table \ref{t:Prevalence} as
input values, following results are obtained with the algorithm described above (Table
\ref{t:results}).

\begin{table*}[ht]
 \centering
 \def\~{\hphantom{0}}
 \begin{minipage}{140mm}
  \caption{Age-specific incidence rates for dementia as calculated with the new method.} \label{t:results}
\begin{tabular*}{\textwidth}{@{}l@{\extracolsep{\fill}}l@{\extracolsep{\fill}}
l@{\extracolsep{\fill}}}%
  \hline \hline
Age group  & Females & Males\\ 
(in years) & (per 100 person-years) & (per 100 person-years) \\ \hline
60--64 & 0.1 & 0.1 \\
65--69 & 0.2 & 0.3 \\
70--74 & 0.6 & 0.5 \\
75--79 & 1.2 & 1.1 \\
80--84 & 2.9 & 2.6 \\
85--89 & 5.4 & 4.8 \\
90--94 & 9.7 & 8.4 \\
\hline
\end{tabular*}
\end{minipage}
\vspace*{-6pt}
\end{table*}

\section{Discussion}\label{sec:discussion}
This paper descibes a new method for deriving incidence rates of a chronic disease from 
prevalence data. It relies on a simple compartment model with three states and 
transitions between these. 
With the assumptions that the transition rates just depend on age $a$ and the 
population is stationary, a one-dimensional ODE relates 
the change in the age-specific prevalence to the incidence and mortality rates. 
After the age stratified prevalence data is transformed into a differentiable form, the 
ODE can be solved for the age-specific incidence rate. For this purpose, the natural cubic
interpolation spline is used. So far, this choice is arbitrary, there might be better
ones.

\bigskip

While the incidence in the system \eqref{eq:MurrayODE} can only be extracted 
with sophisticated methods (for example with a restricted
optimization), the approach based on the ODE \eqref{eq:eq4} and the spline
is considerably less computationally intensive.
Computation time can be a problem, because typically
the prevalence data are fraught with errors and 
a sensitivity analysis should be performed. In such
sensitivity analyses, many (thousands) of constellations
of the input data (i.e. the prevalence in the age groups) are generated
and the changes in the result (age-specific
incidence rates) are monitored. In a validation study of the
method treating data from dialysis patients, 1500
optimizations took more than six hours on an AMD Quad-Core PC with
2.6 GHz.

\bigskip

Because the data of the SHI used for \citep{Zie09} covers a period of one year, 
four reporting periods (quarters) are spanned. Based on this, the authors try to estimate 
the incidence rate, too. When a member of the SHI gets a diagnosis of dementia
in second or third quarter but not in the first quarter, 
\citeauthor{Zie09} consider this as a potentially new case. 
To avoid false positives (dementia in the early stage is difficult to be seen), 
only those cases from the potentially new cases are finally taken into account,
in which the fourth quarter also contained a diagnosis of dementia. 
Using this method, the authors report the incidence rates as shown in Table
\ref{t:Zie09Inz}. For comparison the values of the new ODE method are shown in brackets.

\begin{table*}[ht]
 \centering
 \def\~{\hphantom{0}}
 \begin{minipage}{140mm}
  \caption{Age-specific incidence rates, \citep[Tab. 3]{Zie09}. Values
in brackets are the results of our method (cf. Tab. \ref{t:results}).} \label{t:Zie09Inz}
\begin{tabular*}{\textwidth}{@{}l@{\extracolsep{\fill}}l@{\extracolsep{\fill}}
l@{\extracolsep{\fill}}}%
  \hline \hline
Age group  & Females & Males\\ 
(in years) & (per 100 person-years) & (per 100 person-years)\\ \hline
65--69 & 0.3 (0.2) & 0.3 (0.3) \\
70--74 & 0.8 (0.6) & 0.7 (0.5) \\
75--79 & 1.8 (1.2) & 1.7 (1.1) \\
80--84 & 3.5 (2.9) & 3.0 (2.6) \\
85--89 & 6.9 (5.4) & 5.2 (4.8) \\
90--94 & 9.7 (9.7) & 7.6 (8.4) \\
\hline
\end{tabular*}
\end{minipage}
\vspace*{-6pt}
\end{table*}

Compared with the values of our method, the results in Table \ref{t:Zie09Inz} are, up to
few exceptions, higher. There are two possible reasons:
\begin {enumerate}
 \item \citeauthor{Zie09} overestimate the incidence by their method,
because prevalent cases with no doctor visit in the first quarter count as 
a newly incident case. This is very likely, since incidence estimates relying 
on one disease free quarter only are very prone to overestimations. For a recent
work reflecting on this, see \citep{Abb11}. 
 \item On the other hand, it may be that our values are systematically
too small. One reason might be the increased relative mortality
in the first year after diagnosis of dementia. Instead of the
measured value $R = 3.7 $, here $R = 2.4$ is used. Hence, our relative risk 
of death is too low, which manifests in an underestimation of 
the incidence. However, in comparison with the age-specific
incidences in other studies, \citep[Fig. 3]{Zie09},
our values are in a very good agreement.

\end{enumerate}

%% file: dementia_engl.bbl
\begin{thebibliography}{}
\bibitem[\protect\citeauthoryear{Abbas et~al.}{2011}]{Abb11}
Abbas, S., Ihle, P., Köster, I., Schubert, I. (2011). Estimation of Disease Incidence in 
Claims Data Dependent on the Length of Follow-Up: A Methodological Approach. 
{\it Health Serv Res} doi: 10.1111/j.1475-6773.2011.01325.x. [Epub ahead 
of print]


%\bibitem[\protect\citeauthoryear{Barendregt et~al.}{2003}]{Bar03}
%Barendregt, J. J., van Oortmarssen, G. J., Vos, T., Murray, C. J.
%L. (2003). A generic model for the assessment of disease
%epidemiology: the computational basis of DisMod II. {\it
%Population Health Metrics} {\bf 1 (1)} 4

%\bibitem[\protect\citeauthoryear{Brinks et~al.}{2011}]{Bri11}
%Brinks, R., Landwehr, S., Icks, A., Giani, G. (2011). Deriving the Age-specific 
%Incidence of a Chronic Disease from Prevalence Data {\it
%submitted to Biometrics}

\bibitem[\protect\citeauthoryear{Godfrey}{1983}]{God83}
Godfrey, K. (1983). { \it Compartmental models and their application.} San Diego: Academic Press.

\bibitem[\protect\citeauthoryear{Kamke}{1983}]{Kam83}
Kamke, E. (1983). {\it Differentialgleichungen.} Stuttgart:
Teubner.

\bibitem[\protect\citeauthoryear{Keeling and Rohani}{2007}]{Kee07}
Keeling, M., Rohani, P. (2007). {\it Modeling Infectious Diseases
in Humans and Animals.} Princeton: Princeton University Press.

\bibitem[\protect\citeauthoryear{Keiding et~al.}{1990}]{Kei90}
Keiding, N, Hansen, B. E., Holst, C. (1990). Nonparametric 
Estimation of Disease Incidence from a Cross-Sectional Sample 
of a Stationary Population. {\it Lecture Notes in Biomath} 
{\bf 86} 36--45

\bibitem[\protect\citeauthoryear{Keiding}{1991}]{Kei91}
Keiding, N. (1991). Age-specific incidence and prevalence: a
statistical perspective. {\it Journal of the Royal Statistical
Society A} {\bf 154} 371--412


%\bibitem[\protect\citeauthoryear{Kruijshaar et~al.}{2002}]{Kru02}
%Kruijshaar, M. E., Barendregt, J. J., Hoeymans, N. (2002) The use
%of models in the estimation of disease epidemiology. {\it Bulletin
%of the WHO} {\bf 80 (8)} 622--628


\bibitem[\protect\citeauthoryear{Langohr}{1999}]{Lan99}
Langohr, K. (1999). Estimation of the Incidence of 
Disease with the Use of Prevalence Data 
\verb"https://eldorado.tu-dortmund.de/bitstream/2003/4949/1/99_12.pdf"
Accessed 10.12.2011.

%\bibitem[\protect\citeauthoryear{Lugert}{2007}]{Lug07}
%Lugert, P. (2007). Stichprobendaten von Versicherten 
%der gesetzlichen Krankenversicherung {\it FDZ-Arbeitspapier Nr. 22}, 
%Wiesbaden: Statistisches Bandesamt Forschungsdatenzentrum


\bibitem[\protect\citeauthoryear{Murray and Lopez}{1994}]{Mur94}
Murray, C. J. L. and Lopez, A. D. (1994). Quantifying disability:
data, methods and results {\it Bulletin of the WHO} {\bf 72 (3)}
481--494

\bibitem[\protect\citeauthoryear{Murray and Lopez}{1996}]{Mur96}
Murray, C. J. L. and Lopez, A. D. (1996). Global and regional descriptive
epidemiology of disability: incidence, prevalence, health expectancies and 
years lived with disability. In: Murray, C. J. L., Lopez, A.D. (ed.) 
{\it The Global Burden of Disease}. Boston: Harvard School of Public Health, 201--46.

\bibitem[\protect\citeauthoryear{Preston and Coale}{1982}]{Pre82}
Preston, S. H. and Coale, A. J. (1982). Age structure, growth,
attrition, and accession: a new synthesis, {\it Population Index}
{\bf 48 (2)} 217--59

\bibitem[\protect\citeauthoryear{Preston et~al.}{2001}]{Pre01}
Preston, S. H., Heuveline, P. and Guillot, M. (2001). {\it Demography}. Malden, MA: Blackwell.

\bibitem[\protect\citeauthoryear{Rait et~al.}{2010}]{Rai10}
Rait, G., Walters, K., Bottomley, C. et~al. (2010). Survival of People with Clinical 
Diagnosis of Dementia in Primary Care, {\it British Medical Journal}
{\bf 341} c3584

\bibitem[\protect\citeauthoryear{Szklo and Nieto}{2007}]{Szk07}
Szklo, M., Nieto, F. J. (2007). {\it Epidemiology: Beyond the Basics.} Sudbury, MA: Jones and Bartlett

\bibitem[\protect\citeauthoryear{Ziegler and Doblhammer}{2009}]{Zie09}
Ziegler, U., Doblhammer, G. (2009). Prävalenz und Inzidenz von Demenz in Deutschland, 
{\it Gesandheitswesen} {\bf 71} 281--190

\end{thebibliography}
